\documentclass[letterpaper, 10 pt, conference]{ieeeconf}

\IEEEoverridecommandlockouts
\title{\LARGE \bf
Flatness-Based Geometric Path Following via Guiding Vector Fields*}

\author{J.P. van Gool$^{1}$, H.G. de Marina$^{2}$, B. Jayawardhana$^{1}$, and S. Ahmed$^{1}$ %
\thanks{*This work was supported by the NWO NGF NXTGEN HighTech project ``FARMLAB" under project number NGF.1678.25.012.}
\thanks{$^{1}$Engineering and Technology Institute Groningen, Faculty of Science and Engineering, University of Groningen, Groningen, The Netherlands (e-mail: j.p.van.gool@rug.nl, b.jayawardhana@rug.nl, s.ahmed@rug.nl)}
\thanks{$^{2}$Institute of Mathematics of the University of Granada (IMAG), University of Granada, Granada, Spain (e-mail: hgdemarina@ugr.es)}%
}

\let\clsproof\proof
\let\endclsproof\endproof
\let\proof\relax
\let\endproof\relax
\usepackage{amsthm}
\let\proof\clsproof
\let\endproof\endclsproof
\renewcommand{\QED}{\QEDopen}

\usepackage{amsmath, amssymb}
\usepackage{hyperref}
\usepackage{graphicx}
\usepackage{booktabs}
\usepackage{xfrac}

\newtheorem{proposition}{Proposition}

\theoremstyle{definition}
\newtheorem{definition}{Definition}
\theoremstyle{remark}
\newtheorem{remark}{Remark}
\usepackage{etoolbox}
\AtEndEnvironment{remark}{\hfill$\blacktriangleleft$}

\usepackage{xcolor}

\begin{document}
\maketitle
\thispagestyle{empty}
\pagestyle{empty}

\begin{abstract}
Path following requires an agent to converge to and traverse a geometric path without a prescribed timing law. Guiding vector fields (GVFs) address this directly, but as guidance laws they offer no mechanism to compensate for the agent's own dynamics. Differential flatness-based control (DFBC), on the other hand, compensates for the agent's nonlinear dynamics but relies on a time-parameterized reference trajectory. This letter presents a unified algorithm, which we term \emph{flatness-based geometric control} (FGC), that generates the DFBC reference hierarchy directly from a GVF and its higher-order time derivatives, eliminating the need for a pre-planned time-parameterized trajectory while retaining full dynamics-informed control. A cascade Lyapunov argument establishes exponential convergence of the tracking error and the physical trajectory to the desired path. The framework is further extended to singularity-free GVFs, which removes the topological obstruction and guarantees the condition for global exponential convergence on closed or self-intersecting paths. The proposed algorithm is validated in simulation.
\end{abstract}

\section{INTRODUCTION}
\noindent Autonomous navigation along geometric paths arises in many robotic applications, from infrastructure inspection and surveillance to formation flight and cooperative transport~\cite{chung_survey_2018},~\cite{loianno_special_2020}. For agile platforms such as quadrotors, the challenge is compounded by underactuation and strongly nonlinear dynamics. A controller must simultaneously compensate for internal forces, track higher-order motion references, and remain robust to disturbances~\cite{sun_comparative_2022}. A further complication arises when the mission geometry, rather than a timing law, is the primary constraint. In such settings, the classical trajectory tracking formulation, in which a time-parameterized reference is tracked as closely as possible, is unnecessarily restrictive and its transient behavior depends on the initial condition of the state and that of the reference signal. Path following, where convergence is defined relative to a geometric curve, is the more natural objective~\cite{zhao_integrating_2018}. 

Geometric approaches, such as guiding vector fields (GVFs), address this directly by constructing a vector field whose integral curves converge to and traverse a desired path without imposing any timing law~\cite{yao_guiding_2023}, yielding a purely spatial feedback law with provable convergence. A traditional limitation of conventional GVF is the existence of singular points. This is resolved by~\cite{yao_singularity-free_2021} via so-called singularity-free GVFs, which lift the path into an extended state space via a virtual coordinate, guaranteeing global convergence. Further extensions handle occluded paths and unbounded curves~\cite{yao_guiding_2022},~\cite{ feurgard_extending_2024}. However, GVFs are not stand-alone controllers. They prescribe a velocity direction, and the response to internal dynamics and model nonlinearities falls to a lower-level tracking controller. This leaves a gap for platforms with non-trivial internal forces.

Differential flatness-based control (DFBC) is a natural candidate for closing this gap, since it exploits the algebraic structure of flat systems to cancel nonlinear drift and recover linear, tunable closed-loop error dynamics~\cite{fliess_flatness_1995, mellinger_minimum_2011}. For quadrotors, differential flatness is a standard tool for agile flight~\cite{sun_comparative_2022, poultney_robust_2018}, with current extensions looking at rotor drag and model predictive control~\cite{faessler_differential_2018, greeff_flatness-based_2018}.

However, DFBC does not itself address the path-following problem. It relies on a time-parameterized reference trajectory, and any mismatch between the planned timing and the actual motion introduces avoidable tracking error, particularly for infeasible or aggressive paths~\cite{sun_comparative_2022}.

The two limitations identified above are complementary: DFBC lacks a mechanism for generating a reference without a pre-planned timing law, while GVFs, as guidance laws, do not by themselves provide dynamics-informed control. This letter aims to combine the time-independent nature of geometric path following with the algebraic structure of differentially flat systems, which enables cancellation of nonlinear drift, into a single control law. This combination is motivated by agile flight applications, such as aggressive quadrotor maneuvering or aerobatics, where a mission is naturally specified by geometry rather than a timing law, but where high-speed performance still demands dynamics-informed control rather than kinematic guidance alone.

The combination of the two approaches is not without precedent. One line of work applies the GVF directly at the kinematic level~\cite{zhao_integrating_2018, kapitanyuk_guiding_2018, hu_spontaneous-ordering_2023}, where the guidance law supplies only the field's current value, and any change along the agent's motion is corrected reactively, without compensating for internal dynamics. A tighter integration appears in~\cite{zhou_vector_2014}, which uses differential flatness to map a desired velocity field directly to nonlinear control inputs, enabling the quadrotor to follow a vector field as if it were a simple integrator. There, flatness tracks the field's current value rather than anticipating its change, with no bound on the resulting tracking error, precisely what matters in the aggressive regimes this letter targets.

The direction taken here is the converse. The GVF and its higher-order time derivatives, themselves $k$-th order guiding vector fields, are used to generate the reference that feeds into the DFBC law, anticipating changes in the field rather than correcting for them reactively, while still eliminating the need for a pre-planned time trajectory and retaining full dynamics-informed control within a single closed-loop design. To the best of the authors' knowledge, this approach has not previously been proposed.

The contributions of this letter are as follows: (i) a unified algorithm that uses flatness-based, dynamics-informed control driven by spatially generated references rather than time-parameterized ones, with exponential convergence established via a cascade Lyapunov argument, and (ii) an extension to singularity-free GVFs, which removes the topological obstruction and guarantees the condition for global exponential convergence. Both results are validated in simulation.

The remainder of this letter is structured as follows. Section~\ref{sec:preliminaries} reviews the background on DFBC and GVFs. Section~\ref{sec:algorithm} presents the proposed algorithm, derives the reference hierarchy, and provides the stability and convergence analysis. Section~\ref{sec:simulations} presents numerical simulations, and Section~\ref{sec:conclusion} concludes the letter.

\section{PRELIMINARIES}\label{sec:preliminaries}
\subsection{Notation}
\noindent Throughout the letter, the agent's position is denoted by ~${\xi \in \mathbb{R}^n}$, the full system state by $x \in \mathbb{R}^{n_x}$, and the flat output by $z = h(x) \in \mathbb{R}^m$. The superscript $(k)$ denotes the $k$-th time derivative along system trajectories, e.g., $\chi^{(k)}$, $\xi^{(k)}$, with $\chi^{(0)} \triangleq \chi$. The notation $\mathcal{D}^p\chi$ denotes the $p$-th order derivative of $\chi$ as a symmetric $p$-linear map, and $\mathcal{B}_{k,p}$ denotes the partial Bell polynomial, which sums over all ways of partitioning the derivative order $k$ into $p$ factors~\cite{comtet_advanced_1974}. The symbol $\circ$ denotes the Hadamard product and power. The operator $\nabla_\times\phi(\xi)$ denotes a vector field spanning the null space of the path-error Jacobian $\nabla\phi(\xi)^\top$.

\subsection{Differential flatness-based control}\label{sec:dfbc}
\noindent Consider a nonlinear control-affine system
\begin{equation}
    \dot{x} = f(x) + g(x)\,u,
    \label{eq:system}
\end{equation}
where $x \in \mathbb{R}^{n_x}$ is the state, $u \in \mathbb{R}^m$ is the control input, and $f$, $g$ are smooth maps.

\begin{definition}[Differentially flat system]\label{def:flat}
System~\eqref{eq:system} is {\em differentially flat} if there exists a variable $z \in \mathbb{R}^m$, of the same dimension as the input and with differentially independent components, such that $z = \Lambda(x, u, \dot{u}, \ldots, u^{(j)})$ and, conversely, $x = \Psi_x(z, \ldots, z^{(q-1)})$ and $u = \Psi_u(z, \ldots, z^{(q)})$, for some finite $j, q \in \mathbb{N}$ and smooth maps $\Lambda, \Psi_x, \Psi_u$ defined on open sets. Such a $z$ is called a {\em flat output} \cite{fliess_flatness_1995}.
\end{definition}
Flatness thus allows the state and input to be recovered algebraically from $z$ and its derivatives, so any sufficiently smooth $t \mapsto z(t)$ for which the flat parametrization is well-defined generates an admissible trajectory of~\eqref{eq:system}, a property exploited in DFBC. Throughout this letter, the flat output is a function of the state alone, $z = h(x)$.

For system~\eqref{eq:system} with uniform relative degree $\rho$ with respect to $z$, differentiating $z = h(x)$ along system trajectories yields the compact input-output form:
\begin{equation}
    z^{(\rho)} = F(x) + G(x)\,u,
    \label{eq:io}
\end{equation}
where $F(x) = L_f^{\rho} h(x)$ is the nonlinear drift, i.e., the contribution of the vector field $f$, and $G(x) =L_g L_f^{\rho-1}h(x)$ is the decoupling matrix~\cite{isidori_nonlinear_1995}. Inverting~\eqref{eq:io}, and prescribing the desired $\rho$-th derivative of $z$ via a reference trajectory $z_r$ yields the differential flatness-based control law:
\begin{equation}
    u = \underbrace{G(x)^{-1}\!\left(z_r^{(\rho)} - F(x)\right)}_{\text{feedforward}}
      + \underbrace{G(x)^{-1}\sum_{i=0}^{\rho-1} K_i\left(z_r^{(i)} - z^{(i)}\right)}_{\text{feedback}},
    \label{eq:dfbc}
\end{equation}
where $K_i$ are positive-definite gain matrices driving the tracking error to zero with freely assignable poles. 

If the relative degree does not span the full state, equation~\eqref{eq:io} is complemented by internal dynamics, which are required to be input-to-state stable (ISS) with respect to $(z, \ldots, z^{(\rho-1)})$ for~\eqref{eq:dfbc} to stabilize the full state~\cite{khalil_nonlinear_2002}.

\begin{remark}[Disturbance rejection]
In practice, model uncertainty and disturbances cause imperfect cancellation of $F(x)$. A sensor-based inner loop, such as incremental nonlinear dynamic inversion (INDI), is often augmented around~\eqref{eq:dfbc} for model-independent disturbance rejection~\cite{sun_comparative_2022}. 
\end{remark}

\subsection{Guiding Vector Fields for Path Following}\label{sec:gvf}
\noindent Path following requires an agent to converge to and traverse a desired geometric path without imposing a timing law. One approach is to design a GVF whose integral curves converge to and follow the desired path~\cite{yao_guiding_2023}.

A desired path $\mathcal{P} \subset \mathbb{R}^n$ is implicitly represented as 
\begin{equation*}
    \mathcal{P} = \bigl\{\,\xi \in \mathbb{R}^n : \phi_i(\xi) = 0,\; i = 1, \ldots, n-1\,\bigr\},
\end{equation*}
where each $\phi_i : \mathbb{R}^n \to \mathbb{R}$ is $C^2$, so that $\mathcal{P}$ generically defines a one-dimensional curve via $n-1$ independent constraints. The GVF is constructed as
\begin{equation}
    \chi(\xi) = \chi_{\parallel}(\xi) - \chi_{\perp}(\xi),
    \label{eq:gvf}
\end{equation}
where the propagation term $\chi_{\parallel}(\xi):=\nabla_\times\phi(\xi)$ generates the motion tangentially along $\mathcal{P}$. The converging term
\[
    \chi_{\perp}(\xi) := \sum_{i=1}^{n-1} k_i\,\phi_i(\xi)\,\nabla\phi_i(\xi), \qquad k_i > 0,
\]
vanishes on $\mathcal{P}$, rendering the path invariant. For the weighted constraint error $V := \frac{1}{2}\sum_{i=1}^{n-1} k_i \phi_i^2$, the dynamics $\dot{\xi} = \chi$ satisfy $\dot{V} = -\|\chi_{\perp}\|^2 \leq 0$, since $\nabla\phi^\top \chi_{\parallel} = 0$. Only the direction of $\chi$ matters for this guarantee, since a positive rescaling preserves its integral curves up to time reparameterization. As $\|\chi\|$ is not constant along $\mathcal{P}$, the normalized GVF $\hat{\chi} := \chi/\|\chi\|$ may be used in its place, wherever $\chi \neq 0$, so that $\dot{\xi} = v\hat{\chi}(\xi)$ traverses $\mathcal{P}$ at the constant speed $v > 0$~\cite{kapitanyuk_guiding_2018}.

A limitation of GVFs is the possible existence of singular points, where $\chi(\xi) = 0$. For a single closed path in $\mathbb{R}^2$, which is a periodic orbit of the $C^1$ field $\chi$, a corollary of the Poincar\'{e}--Bendixson theorem guarantees at least one equilibrium inside the enclosed region~\cite{khalil_nonlinear_2002}, making singularities topologically unavoidable in this setting~\cite{yao_singularity-free_2021}.

To overcome this obstruction, the state is augmented with a virtual coordinate $w \in \mathbb{R}$ as in~\cite{yao_singularity-free_2021}. Given a parameterization $\xi_i = f_i(w)$, the path is lifted to $(\xi, w) \in \mathbb{R}^{n+1}$ through the $n$ implicit constraints $\tilde{\phi}_i(\xi,w) = \xi_i - f_i(w)$. The lifted path $\tilde{\mathcal{P}}$ is unbounded in $w$ and therefore non-self-intersecting, so it is not a periodic orbit and the topological obstruction disappears. A singularity-free GVF $\tilde{\chi}(\xi,w)$ is constructed following the same principles:
\begin{equation}
    \dot{\xi} = \tilde{\chi}_\xi(\xi, w), \qquad \dot{w} = \tilde{\chi}_w(\xi, w).
    \label{eq:gvf_ext}
\end{equation}
It is shown in~\cite{yao_singularity-free_2021} that $\tilde{\chi}(\xi,w) \neq 0$ for all $(\xi, w) \in \mathbb{R}^{n+1}$, guaranteeing global convergence of $\xi(t)$ to $\mathcal{P}$.

\section{PROPOSED ALGORITHM: FLATNESS-BASED GEOMETRIC PATH FOLLOWING}\label{sec:algorithm}
\noindent The DFBC framework in Section~\ref{sec:dfbc} assumes a time-parameterized reference trajectory $z_r(t)$, while the GVF framework in Section~\ref{sec:gvf} provides a stand-alone kinematic path-following law. The proposed method combines the two: rather than supplying a pre-specified $z_r(t)$, the reference velocity is generated by the guiding vector field evaluated at the agent's current position. This replaces temporal trajectory tracking with spatial path following, while inheriting the convergence guarantees of the GVF and retaining the model-inversion structure of DFBC. We refer to the resulting method as \emph{flatness-based geometric control} (FGC).

FGC is advantageous specifically when the mission specifies a geometric path without a required arrival-time profile, e.g., aerobatic maneuvers, inspection sweeps, or racing lines, where any feasible traversal is acceptable. Such paths often demand extreme motion or aggressive speed changes, for which a stand-alone GVF's kinematic guidance is insufficient and dynamics-informed compensation shows its merit.

Consider system~\eqref{eq:system} with flat output $z = h(x) = \xi$ and uniform relative degree $\rho$. As established in~\eqref{eq:dfbc}, the DFBC law requires a reference trajectory and its time derivatives up to order $\rho$. The GVF prescribes the desired flat-output velocity by setting $\dot{z}_r := \chi$. Higher-order references follow by successive differentiation along system trajectories,
\[
    z_r^{(k)} = \chi^{(k-1)}, \qquad k = 1, \ldots, \rho.
\]
Since the relative degree ensures that $u$ does not appear in $\xi^{(k)}$ for $k < \rho$, each derivative $\chi^{(k)}$ is computable from the current state $x$ alone. Substituting the GVF reference hierarchy into~\eqref{eq:dfbc} gives
\begin{equation}\label{eq:combined_law}
u = G(x)^{-1}\left(\chi^{(\rho-1)} - F(x)
  + \sum_{i=1}^{\rho-1} K_i\big(\chi^{(i-1)} - \xi^{(i)}\big)\right).
\end{equation}
The sum starts at $i=1$ rather than $i=0$, since the GVF prescribes a velocity reference $\dot{z}_r = \chi$ rather than a position reference. Path convergence is recovered separately via a cascaded argument on the resulting error dynamics, discussed later. The feedforward term $\chi^{(\rho-1)}$ encodes the spatial variation of the GVF along the current direction of motion.

The references $\chi^{(k)}$ required by~\eqref{eq:combined_law} are time derivatives of $\chi(\xi(t))$, where $\xi$ evolves along the system trajectories. Since $\chi$ depends on $t$ only through $\xi$, each derivative follows from the multivariate chain rule via the Fa\`a di Bruno formula~\cite{johnson_curious_2002}:
\begin{equation}\label{eq:faa_di_bruno}
    \chi^{(k)} = \sum_{p=1}^{k}
    \mathcal{D}^p\chi \cdot
    \mathcal{B}_{k,p}\!\left(\xi^{(1)}, \ldots, \xi^{(k-p+1)}\right).
\end{equation}
\begin{remark}[Smoothness Requirement]\label{rem:smoothness}
Formula~\eqref{eq:faa_di_bruno} requires $\chi \in C^\rho$ in $\xi$, which in turn requires $\phi \in C^{\rho+1}$. Algebraic level sets such as circles or lemniscates satisfy this automatically. For paths specified via splines, sextic, or higher construction guarantees $C^5$ continuity, sufficient for systems up to relative degree $\rho = 4$. Sinusoidal definitions are $C^\infty$ and impose no practical constraints. For piecewise or non-differentiable paths, such as linearly interpolated waypoints, mollification~\cite{gonzalez-calvin_efficient_2026} yields a $C^\infty$ smooth approximation, satisfying this requirement for any $\rho$.
\end{remark}
Identifying $\mathcal{D}^1\chi = J_\chi$ (Jacobian), $\mathcal{D}^2\chi = \mathcal{H}_\chi$ (Hessian), and $\mathcal{D}^3\chi = \mathcal{T}_\chi$ (third-order tensor), the first three instances expand as
\[
\begin{aligned}
\chi^{(1)} &= J_\chi\,\xi^{(1)}, \quad
\chi^{(2)} = J_\chi\,\xi^{(2)} + \mathcal{H}_\chi\!\left(\xi^{(1)},\xi^{(1)}\right), \\[4pt]
\chi^{(3)} &= J_\chi\,\xi^{(3)} + 3\,\mathcal{H}_\chi\!\left(\xi^{(1)},\xi^{(2)}\right)
            + \mathcal{T}_\chi\!\left(\xi^{(1)},\xi^{(1)},\xi^{(1)}\right).
\end{aligned}
\]
The leading term $p=1$ recovers the Jacobian. The remaining terms involve only $\xi^{(1)},\ldots,\xi^{(k-1)}$ and are thus computable from the current state. Each differentiation adds exactly one tensor of $\chi$ and one state derivative $\xi^{(k)}$, reflecting the system's integrator structure.

\subsection{Stability and Convergence Analysis}\label{sec:convergence}
\noindent Convergence is proven in two steps, first showing convergence of the tracking error, and then the consequent convergence of the physical trajectory to the desired path.

Define the error chain
\begin{equation}\label{eq:error_chain}
    e_i := \chi^{(i-1)} - \xi^{(i)}, \qquad i = 1, \ldots, \rho-1,
\end{equation}
so that $e_1 = \chi - \dot{\xi}$ is the velocity error and $e_2 = \dot{\chi} - \ddot{\xi}$ the acceleration error and so on. Differentiating gives $\dot{e}_i = e_{i+1}$ for $i = 1,\ldots,\rho-2$. For the last error, substituting~\eqref{eq:combined_law} into $\dot{e}_{\rho-1} = \chi^{(\rho-1)} - \xi^{(\rho)}$ and using~\eqref{eq:io} yields
\[
    \dot{e}_{\rho-1} = -\sum_{i=1}^{\rho-1} K_i\,e_i.
\]
The error dynamics are therefore the linear time-invariant system
\begin{equation}
\label{eq:ErrorMat}
\dot{e} = A_e\,e
\end{equation}
with characteristic polynomial $s^{\rho-1} + K_{\rho-1}s^{\rho-2} + \cdots + K_1 = 0$. Choosing $K_1,\ldots,K_{\rho-1}$ such that $A_e$ is Hurwitz renders $e = 0$ globally exponentially stable~\cite{khalil_nonlinear_2002}, i.e.,
\begin{equation}\label{eq:e_bound}
    \|e(t)\| \leq c\,\|e(0)\|\,e^{-\lambda t}, \qquad \forall\,t \geq 0,
\end{equation}
for some $c \geq 1$ and $\lambda > 0$.

The bound~\eqref{eq:e_bound} establishes that $e_1 = \chi(\xi) - \dot{\xi} \to 0$ exponentially, meaning the agent velocity converges to the GVF. Convergence to $\mathcal{P}$ must follow from the structure of $\chi$ itself. Writing $\dot{\xi} = \chi - e_1$ and differentiating $\phi(\xi(t))$,
\[
    \dot{\phi} = \nabla\phi^\top\chi - \nabla\phi^\top e_1.
\]
Substituting $\chi = \chi_\parallel - k_e\phi\,\nabla\phi$ from~\eqref{eq:gvf} and using the tangency $\nabla\phi^\top\chi_\parallel = 0$ gives
\begin{equation}\label{eq:phi_dynamics}
    \dot{\phi} = \underbrace{-k_e\,\phi\,\|\nabla\phi\|^2}_{\text{stabilizing}} \underbrace{-\,\nabla\phi^\top e_1}_{\text{perturbation}}.
\end{equation}
This reveals a cascade structure: when $e_1 = 0$ the dynamics reduce to $\dot{\phi} = -k_e\phi\|\nabla\phi\|^2$, driving $\phi \to 0$ exponentially. The perturbation vanishes exponentially by~\eqref{eq:e_bound}.

\begin{proposition}[Convergence to the Path]\label{prop:path_conv}
Consider system~\eqref{eq:system} under the control law~\eqref{eq:combined_law}, with constant gain matrices $K_1,\dots,K_{\rho-1}$ chosen so that matrix $A_e$ in \eqref{eq:ErrorMat} is Hurwitz. Assume that the internal dynamics are ISS whenever~${m\rho < n_x}$. Let $x(\cdot):[0,\infty)\rightarrow \mathbb{R}^{n_x}$ be any solution of the closed-loop system remaining in a domain where the input--output representation~\eqref{eq:io} is valid and $G(x(t))$ is nonsingular. If there exists $c_{\min} > 0$ such that
\begin{equation}
\label{eq:grad_cond}
    \|\nabla\phi(\xi(t))\| \geq c_{\min} \qquad \text{for all } t \geq 0,
\end{equation}
then $\phi(\xi(t)) \to 0$ exponentially as $t \to \infty$.
\end{proposition}
\begin{proof}
Let $V_\phi = \tfrac{1}{2}\phi^2$. Bounding the perturbation term in~\eqref{eq:phi_dynamics} by Cauchy--Schwarz and applying Young's inequality~\cite{khalil_nonlinear_2002} with $a = |\phi|\,\|\nabla\phi\|$ and $b = \|e_1\|$ absorbs it into the stabilizing term, giving
\[
    \dot{V}_\phi \leq -\beta\,V_\phi + \gamma\,\|e_1(t)\|^2,
\]
with $\beta = 2(k_e - \varepsilon)c_{\min}^2 > 0$ and $\gamma = 1/(4\varepsilon)$ for any $\varepsilon \in (0, k_e)$. By the comparison lemma~\cite{khalil_nonlinear_2002}:
\[
    V_\phi(t) \leq e^{-\beta t}V_\phi(0) + \gamma\!\int_0^t\! e^{-\beta(t-s)}\|e_1(s)\|^2\,\mathrm{d}s.
\]
Since $\|e_1(t)\|^2 \leq c^2\|e(0)\|^2 e^{-2\lambda t}$ by~\eqref{eq:e_bound}, evaluating the integral (for $\beta \neq 2\lambda$) gives
\begin{equation}\label{eq:phi_bound}
    V_\phi(t) \leq \bigl(V_\phi(0) - C\bigr)e^{-\beta t} + C\,e^{-2\lambda t}, \quad C \triangleq \frac{\gamma\,c^2\|e(0)\|^2}{\beta - 2\lambda},
\end{equation}
{so $\phi(\xi(t)) \to 0$ exponentially with rate $\min(\beta/2, \lambda)$. The argument extends to $n-1$ constraints with common gain $k_i = k_e$ and $V_\phi = \tfrac12\|\phi\|^2$, replacing $\|\nabla\phi\|$ by the smallest singular value of the path-error Jacobian $N = [\nabla\phi_1, \ldots, \nabla\phi_{n-1}]$, i.e. $\sigma_{\min}(N) \geq c_{\min} > 0$ along the path.}
\end{proof}
\begin{remark}[Non-Vanishing Gradient]
The assumption $\|\nabla\phi\| \geq c_{\min} > 0$ is the key geometric hypothesis of Proposition~\ref{prop:path_conv}, and fails wherever $\nabla\phi$ vanishes. For closed paths such points are unavoidable with a standard GVF, motivating the singularity-free extension: lifting the path via a virtual coordinate guarantees $\|\nabla\tilde{\phi}\| > 0$ everywhere, so the condition holds globally.
\end{remark}

\subsection{Extension to the Singularity-Free GVF}\label{sec:sf_gvf}
\noindent In the singularity-free formulation~\cite{yao_singularity-free_2021}, the agent state is lifted to the augmented state $\zeta = (\xi^\top, w)^\top \in \mathbb{R}^{n+1}$,  where $w \in \mathbb{R}$ is the virtual coordinate evolving as $\dot{w} = \tilde{\chi}_w(\xi,w)$ per~\eqref{eq:gvf_ext}.

The combined control law~\eqref{eq:combined_law} carries over with $\chi^{(k)}$ replaced by $\tilde{\chi}^{(k)}$, computed via the Fa\`a di Bruno formula applied to the augmented state,
\begin{equation}\label{eq:faa_di_bruno_sf}
    \tilde{\chi}^{(k)}
    = \sum_{p=1}^{k}
      \mathcal{D}^p\tilde{\chi} \cdot
      \mathcal{B}_{k,p}\!\left(\zeta^{(1)}, \ldots, \zeta^{(k-p+1)}\right),
\end{equation}
where $\mathcal{D}^p\tilde{\chi}$ is now the $p$-th order derivative of $\tilde{\chi}$ with respect to $\zeta \in \mathbb{R}^{n+1}$. The leading term recovers
\[
    \tilde{\chi}^{(1)} = J_{\tilde{\chi}}\,\dot{\zeta}
    = \frac{\partial\tilde{\chi}}{\partial\xi}\,\dot{\xi}
    + \frac{\partial\tilde{\chi}}{\partial w}\,\dot{w},
\]
confirming that the virtual coordinate contribution enters naturally through the extended Jacobian, with all higher-order terms generated automatically by the Bell polynomials.

\begin{proposition}[Singularity-Free Convergence]\label{prop:sf_conv}
Consider the setting of Proposition~\ref{prop:path_conv} with the singularity-free GVF~\eqref{eq:gvf_ext} and $\chi^{(k)}$ in the control law~\eqref{eq:combined_law} replaced by $\tilde{\chi}^{(k)}$ from~\eqref{eq:faa_di_bruno_sf}. Then the gradient condition~\eqref{eq:grad_cond} holds globally with ${c_{\min} = 1}$. Moreover, for any solution  \({(x(\cdot),w(\cdot)):[0,\infty)\to\mathbb R^{n_x}\times\mathbb R}\)  of the augmented closed-loop system such that \(x(t)\) remains in a domain where the input–output representation~\(\eqref{eq:io}\) is valid and \(G(x(t))\) is nonsingular for all \(t\geq0\), $\tilde{\phi}(\xi(t), w(t)) \to 0$ exponentially as $t \to \infty$.
\end{proposition}
\begin{proof}
The physical error chain $\tilde{e}_i = \tilde{\chi}_\xi^{(i-1)} - \xi^{(i)}$ obeys the same linear cascade as before, since the control law~\eqref{eq:combined_law} assigns these poles independently of the reference, and~\eqref{eq:faa_di_bruno_sf} supplies the exact derivatives of $\tilde{\chi}_\xi$ along the closed-loop trajectory, differentiating through the virtual coordinate $w$. Hence $\tilde{e}(t) \to 0$ exponentially. For the path error, stacking the lifted constraint gradients gives the augmented path-error Jacobian $\tilde{N} = [\,I_n;\, -f'(w)^\top\,] \in \mathbb{R}^{(n+1)\times n}$, so $\tilde{N}^\top\tilde{N} = I_n + f'(w)f'(w)^\top \succeq I_n$ and hence $\sigma_{\min}(\tilde{N}) \geq 1$ for all $\zeta \in \mathbb{R}^{n+1}$. The hypothesis of Proposition~\ref{prop:path_conv} therefore holds globally with $c_{\min} = 1$, and it applies in the augmented space to give $\tilde{\phi}(\zeta(t)) \to 0$ exponentially. Since $\tilde{\phi} \to 0$ implies $\xi \to f(w) \in \mathcal{P}$, the physical trajectory $\xi(t)$ converges exponentially to $\mathcal{P}$.
\end{proof}
The smoothness requirement of Remark~\ref{rem:smoothness} applied to $\tilde{\chi}$ as a function of $\zeta$ translates directly to $f_i \in C^{\rho+1}$ for the path parameterization $\xi_i = f_i(w)$.

\section{NUMERICAL SIMULATIONS}\label{sec:simulations}
\noindent MATLAB simulations validate the proposed algorithms across three scenarios: verifying the convergence bound, testing internal dynamics, and the singularity-free extension.

\subsection{Circular Path With Radial Drift}\label{sec:sim_gvf}
\noindent Consider the double integrator in $\mathbb{R}^2$:
\begin{equation}\label{eq:exmp_sys}
        \dot{\xi} = \nu, \quad
    \dot{\nu} = f(\xi) + u, \quad
    \xi,\nu,u \in \mathbb{R}^2,
\end{equation}
with radial drift $f(\xi) = \,\xi/\|\xi\|$. The flat output is $z = \xi$, $\rho = 2$, and $x = (\xi,\nu) \in \mathbb{R}^4$, giving $F(\xi) = f(\xi)$ and $G(\xi) = I_2$. The agent follows a circle of radius $R = 5$, $\phi(\xi) = \xi_1^2 + \xi_2^2 - R^2$, with guiding vector field and Jacobian:
\[
\begin{aligned}
        \chi(\xi) &= \begin{bmatrix} -2\xi_2 \\ 2\xi_1 \end{bmatrix}
        - k_e\,\phi(\xi)\begin{bmatrix} 2\xi_1 \\ 2\xi_2 \end{bmatrix},\\
    J_\chi(\xi) &= \begin{bmatrix} 0 & -2 \\ 2 & 0 \end{bmatrix}
        - k_e\!\left(\nabla\phi\,\nabla\phi^\top + 2\phi\,I_2\right).
\end{aligned}
\]
Since $\chi$ depends on $\xi$ alone, the first-order reference derivative is exact, $\chi^{(1)}(\xi,\nu) = J_\chi(\xi)\,\nu$ (first term of~\eqref{eq:faa_di_bruno}), and the $\rho = 2$ instance of~\eqref{eq:combined_law} reduces to
\[
    u = J_\chi(\xi)\,\nu - F(\xi) + K_1\bigl(\chi(\xi) - \nu\bigr),
\]
simulated with $k_e = 0.5$, $K_1 = 3$, and $\xi_0 = [6,\,0]^\top$.

The velocity tracking error $e_1$ decays exponentially at a constant rate, per~\eqref{eq:e_bound}. Figure~\ref{fig:combined_traj} shows the trajectory converging to the circle and traversing it despite the drift. The path error $\phi(\xi)$ in Figure~\ref{fig:combined_error} also decays exponentially, as guaranteed by Proposition~\ref{prop:path_conv}, but is visibly nonlinear: the cascaded structure drives $\phi$ both by its own stabilizing dynamics and by the perturbation $e_1$. The error nonetheless remains within the envelope predicted by~\eqref{eq:phi_bound}.

\begin{figure}[t]
    \centering
    \includegraphics[width=.95\linewidth]{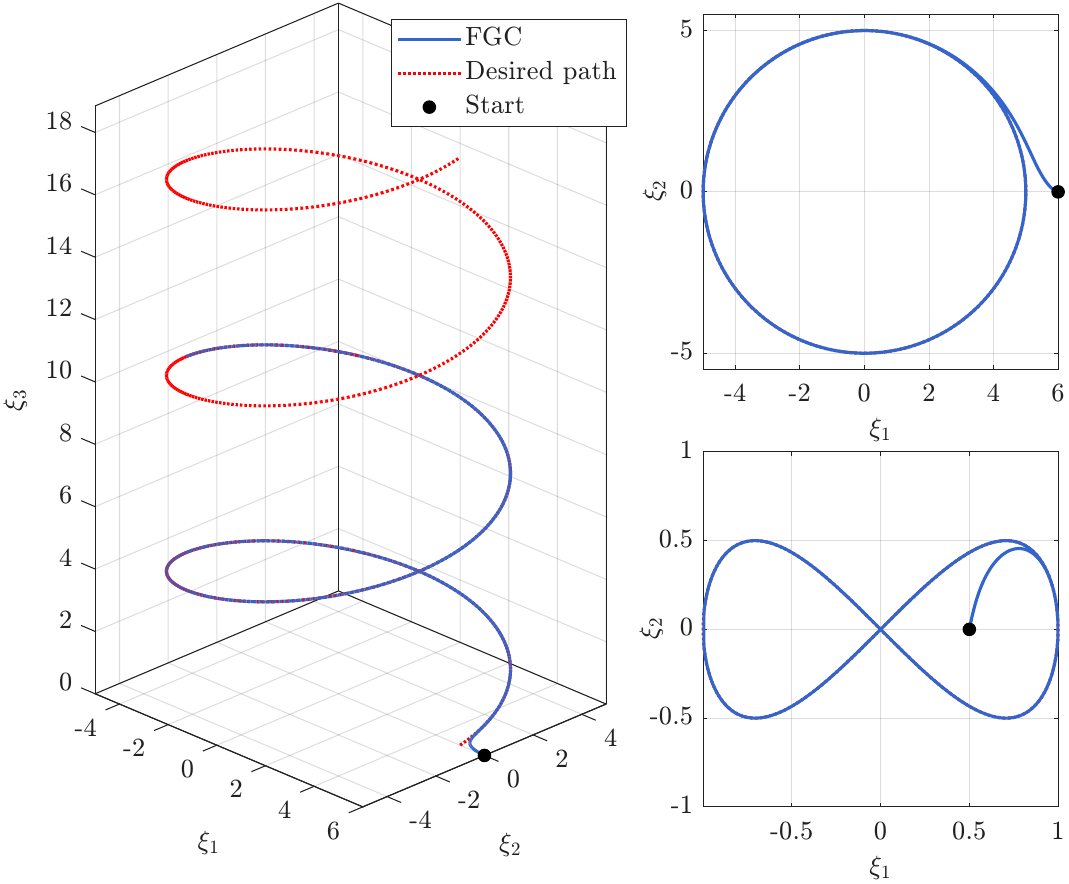}
    \caption{Agent trajectories (solid) converging to and traversing their desired paths (dashed) under FGC. Left: a helical path under drift induced by a coupled internal state. Top right: a circular path under radial drift. Bottom right: a self-intersecting figure-eight path using the singularity-free extension, traversing the crossing at the origin without stalling.}
    \vspace{-1em}
    \label{fig:combined_traj}
\end{figure}

\subsection{Helical Path with Internal Dynamics}
\noindent Extending to $\mathbb{R}^3$ and coupling in a state-dependent internal state, consider
\[
    \dot{\xi} = \nu, \quad
    \dot{\nu} = -\eta \circ \xi + u, \quad
    \dot{\eta} = -\eta + \xi^{\circ 2},
    \quad \xi,\nu,\eta,u \in \mathbb{R}^3.
\]
The flat output is $z = \xi$ with $\rho = 2$ and full state $x = (\xi,\nu,\eta) \in \mathbb{R}^9$, giving $F(x) = L_f^2h(x) = -\eta\circ\xi$ and $G(x) = L_gL_fh(x) = I_3$. On the zero surface $\xi = 0,\ \nu = 0$, the internal dynamics reduce to $\dot\eta = -\eta$, which is globally exponentially stable, confirming the cascade is minimum phase. The path is a helix of radius $R = 5$ and pitch $c = 1$, represented implicitly by
\[
    \phi_1(\xi) = \xi_1^2 + \xi_2^2 - R^2, \quad
    \phi_2(\xi) = \xi_1 \sin\!\left(\tfrac{\xi_3}{c}\right)
                - \xi_2 \cos\!\left(\tfrac{\xi_3}{c}\right),
\]
the first constraining the agent to a cylinder of radius $R$, the second fixing its angular phase along the cylinder to trace a helix of pitch $c$. With $g_i = \nabla\phi_i$ s.t. $g_1 = \begin{bmatrix} 2\xi_1 & 2\xi_2 & 0 \end{bmatrix}^\top$, ~${g_2 = \begin{bmatrix} \sin(\frac{\xi_3}{c}) & -\cos(\frac{\xi_3}{c}) & \bigl(\xi_1\cos(\frac{\xi_3}{c}) + \xi_2\sin(\frac{\xi_3}{c})\bigr)/c \end{bmatrix}^\top}$, the propagation direction is $\tau(\xi) = g_2 \times g_1$, spanning the null space of $\begin{bmatrix} g_1^\top & g_2^\top \end{bmatrix}$, and the guiding vector field is
\[
    \chi(\xi) = \gamma\,\frac{\tau(\xi)}{\|\tau(\xi)\|}
        - k_e\bigl(\phi_1(\xi)\,g_1 + \phi_2(\xi)\,g_2\bigr).
\]
Then, $J_\chi$ follows by differentiation, and the control law is
\[
    u = J_\chi(\xi)\,\nu - F(x) + K_1\bigl(\chi(\xi) - \nu\bigr),
\]
with $k_e = 0.5$, $\gamma = 2$, $K_1 = 8$, and $\xi_0 = [6,\,0,\,0]^\top$.

As shown in Figure~\ref{fig:combined_traj}, the agent converges and follows the helix despite the drift induced by the internal state $\eta$. Rather than the pronounced nonlinear profile seen for the circle, Figure~\ref{fig:combined_error} shows a path error that decays nearly linearly on the logarithmic scale, a result of the interplay between the two path constraints. The ascending component $\phi_2$ dominates the circular component $\phi_1$ and sets the effective convergence rate. The $\eta$ subsystem satisfies the input-to-state stability condition, confirming the law remains applicable when the relative degree does not span the full state.

\begin{figure}[ht]
    \centering
    \includegraphics[width=1\linewidth]{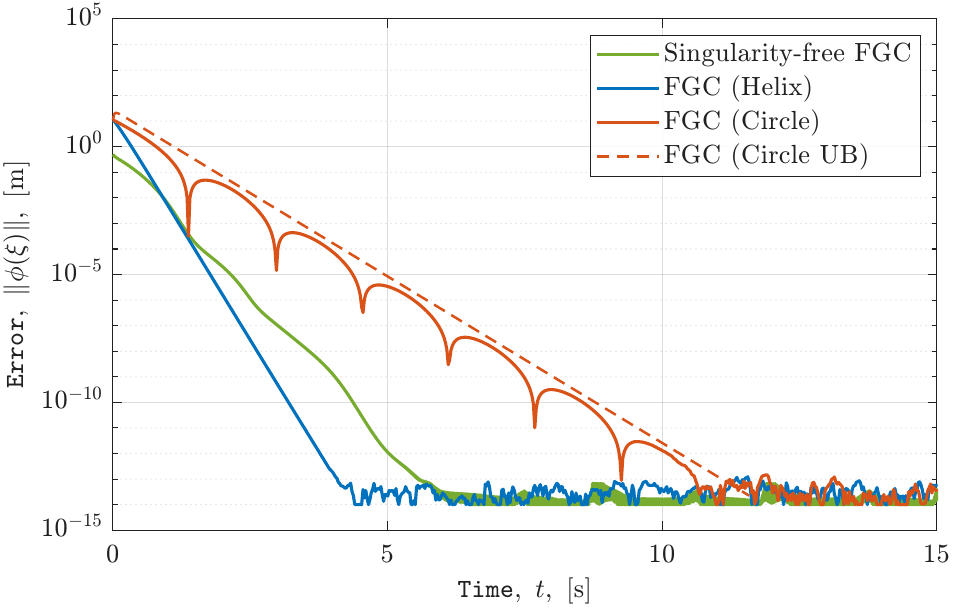}
    \caption{Path-following error $\|\phi(\xi(t))\|$ on a semi-logarithmic
    scale for the helical (blue), circular (orange), and self-intersecting
    figure-eight (green) paths, with the corresponding theoretical
    convergence bound (dashed orange) overlaid for the circular case. All
    three cases exhibit exponential convergence to the desired path.}
    \label{fig:combined_error}
\end{figure}
\vspace{-1em}

\subsection{Singularity-Free FGC}\label{sec:sim_sf}
\noindent Consider a system of type \eqref{eq:exmp_sys}, subject to a Coriolis-like drift $f(\xi) = \beta \begin{bmatrix} \xi_2 & -\xi_1 \end{bmatrix}^\top$ of magnitude $\beta = 0.3$. The target is a self-intersecting figure-eight, which admits a singularity under the standard GVF, so the singularity-free construction is used. The curve is parameterized by the virtual coordinate $w \in \mathbb{R}$ via
$
    f_1(w) = R\sin w, \, f_2(w) = R\sin w \cos w,
$
with $R = 1$, and lifted to the extended state $\zeta = (\xi,w) \in \mathbb{R}^3$ through the implicit constraints
$
    \tilde\phi_i(\xi,w) = \xi_i - f_i(w), \, i = 1,2.
$
Following the same converging/propagating split as~\eqref{eq:gvf}, the singularity-free field is
\[
\begin{aligned}
        \tilde\chi_\xi(\zeta) &= \begin{bmatrix}
        f_1'(w) - k_1\,\tilde\phi_1 \\ f_2'(w) - k_2\,\tilde\phi_2
    \end{bmatrix},\\
    \tilde\chi_w(\zeta) &= 1 + k_1\,\tilde\phi_1 f_1'(w) + k_2\,\tilde\phi_2 f_2'(w).
\end{aligned}
\]
The leading-order reference derivative follows Section~\ref{sec:sf_gvf} with \(\tilde\chi^{(1)}(\zeta) = J_{\tilde\chi}(\zeta)\,\dot\zeta,\) with physical component $\tilde\chi_\xi^{(1)}(\zeta)$. The $\rho = 2$ instance of~\eqref{eq:combined_law} and virtual coordinate update read
\[
    u = \tilde\chi_\xi^{(1)}(\zeta) - F(\xi)
        + K_1\bigl(\tilde\chi_\xi(\zeta) - \nu\bigr), \quad
    \dot{w} = \tilde\chi_w(\zeta),
\]
with $k_1 = k_2 = k_e = 4$, $K_1 = 12$, and $\xi_0 = [0.5,\,0]^\top$.

Figure~\ref{fig:combined_traj} shows the agent converging onto and traversing the figure-eight, including the self-intersection at the origin. Figure~\ref{fig:fig8_extended} shows the same trajectory in the extended space $(\xi_1,\xi_2,w)$, where the lifted path is non-self-intersecting and the agent converges onto it monotonically in $w$, resolving the crossing at $\xi = 0$ through the virtual coordinate. Exponential convergence is thus achieved on a self-intersecting path, which the standard GVF construction cannot deliver. As before, Figure~\ref{fig:combined_error} shows nonlinear path-error behavior following from the same cascaded argument.

\begin{figure}[ht]
    \centering
    \includegraphics[width=1\linewidth]{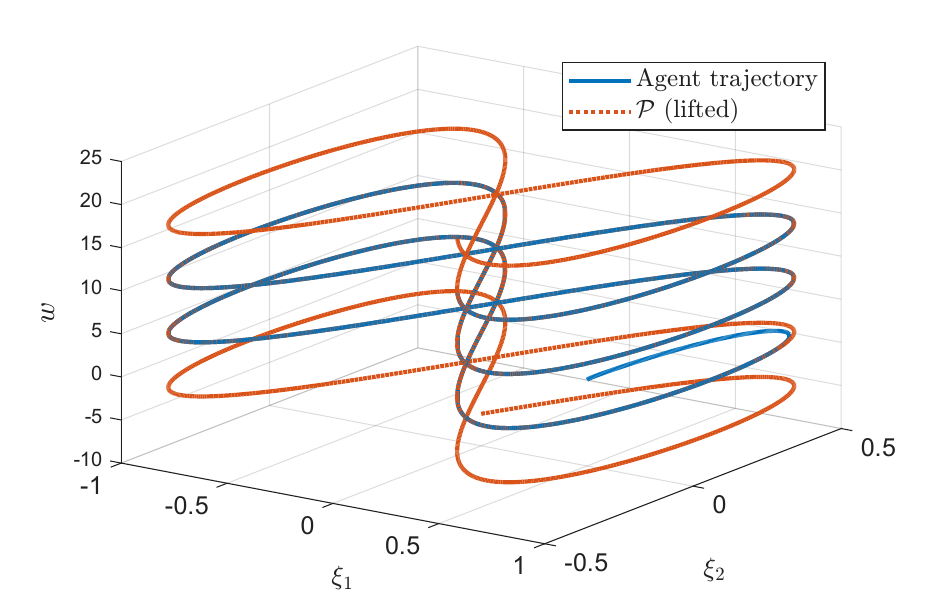}
    \caption{Agent trajectory in the extended state space $(\xi_1,\xi_2,w)$
    under the singularity-free construction. The lifted reference path
    (dashed) is non-self-intersecting, and the agent's trajectory (solid)
    converges onto it, resolving the crossing at $\xi = 0$ through the
    virtual coordinate $w$.}
    \label{fig:fig8_extended}
\end{figure}

\section{CONCLUSION}\label{sec:conclusion}
\noindent This letter presented \emph{flatness-based geometric control} (FGC), a control law that generates the DFBC reference hierarchy directly from a guiding vector field, rather than a pre-planned trajectory. A cascade Lyapunov argument established exponential convergence to the desired path, with the gradient condition made global via singularity-free guiding vector fields. Simulations confirmed these results on helical, circular, and self-intersecting figure-eight paths.

\bibliographystyle{IEEEtran}
\bibliography{references}
\end{document}